\documentclass[12pt]{iopart}
\usepackage{xcolor}
\usepackage{cite}
\usepackage[colorlinks=true,allcolors=blue]{hyperref}
\expandafter\let\csname equation*\endcsname\relax 
\expandafter\let\csname endequation*\endcsname\relax 
\usepackage{amsmath,amsthm,amsfonts,amssymb}

\newtheorem{theorem}{Theorem}
\newtheorem{remark}[theorem]{Remark}

\newtheorem{proposition}[theorem]{Proposition}

\newtheorem{conjecture}[theorem]{Conjecture}

\newcommand{\matrixform}[1]{\ensuremath{\mathcal{#1}}}
\newcommand{\functional}[1]{\ensuremath{\mathtt{#1}}}
\newcommand{\functionspace}[1]{\ensuremath{\mathbb{#1}}}

\newcommand{\kk}{\ensuremath{k}}
\newcommand{\xx}{\ensuremath{x}}
\newcommand{\rr}{\ensuremath{r}}

\newcommand{\dd}{\ensuremath{\mathrm{d}}}

\newcommand{\dx}{\ensuremath{\,\dd\xx}}
\newcommand{\dk}{\ensuremath{\,\dd\kk}}
\newcommand{\dr}{\ensuremath{\,\dd\rr}}

\newcommand{\de}{\ensuremath{\delta}}

\newcommand{\dom}{\functionspace{D}}   
\newcommand{\HH}{\functionspace{H}}   
\newcommand{\RR}{\functionspace{R}} 

\newcommand{\fa}{\functional{A}}
\newcommand{\fb}{\functional{B}}
\newcommand{\fc}{\functional{C}}   
\newcommand{\ff}{\functional{F}}
\newcommand{\fg}{\functional{G}}   
\newcommand{\fh}{\functional{H}}   
\newcommand{\fj}{\functional{J}}   
\newcommand{\fk}{\functional{K}}   
\newcommand{\fs}{\functional{S}}   

\newcommand{\MJ}{\matrixform{J}}  
\newcommand{\MN}{\matrixform{N}}  
\newcommand{\MT}{\matrixform{T}}  
\newcommand{\MU}{\matrixform{U}}  
\newcommand{\MV}{\matrixform{V}}  
\newcommand{\MW}{\matrixform{W}}  
\newcommand{\MX}{\matrixform{X}}  
\newcommand{\MY}{\matrixform{Y}}  

\newcommand{\cbra}[2]{\ensuremath{\left[  #1 , #2 \right] }}   
\newcommand{\pbra}[2]{\ensuremath{\left\{ #1 , #2 \right\}}}   
\newcommand{\mbra}[2]{\ensuremath{\left(  #1 , #2 \right) }}   

\newcommand{\trace}{\ensuremath{\mathrm{tr} \,}}

\newcommand{\fracd}[2]{\ensuremath{\frac{\delta   #1}{\delta   #2}}}
\newcommand{\fracp}[2]{\ensuremath{\frac{\partial #1}{\partial #2}}}
\newcommand{\fracdsqr}[2]{\ensuremath{\frac{\delta^2   #1}{\delta #2^2}}}

\newcommand{\changecolor}[1]{\hypersetup{linkcolor=#1}}  
\begin{document}

\title[A metriplectic formulation of polarized radiative transfer]{A metriplectic formulation of polarized radiative transfer}

\author{V. Bosboom$^1$, M. Kraus$^{2,3}$ and M. Schlottbom$^1$}
\address{$^1$ Department of Applied Mathematics, University of Twente,
P.O. Box 217, 7500 AE Enschede, The Netherlands.}
\address{$^2$ Max-Planck-Institut f\"ur Plasmaphysik
Boltzmannstra\ss e 2, 85748 Garching, Germany.}
\address{$^3$ Technische Universit\"at M\"unchen, Zentrum Mathematik, Boltzmannstra\ss{}e 3, 85748 Garching, Germany}%
\eads{\mailto{v.bosboom@utwente.nl}, \mailto{michael.kraus@ipp.mpg.de}, \mailto{m.schlottbom@utwente.nl}}
\begin{abstract}
    We present a metriplectic formulation of the radiative transfer equation with polarization and varying refractive index and show that this formulation automatically satisfies the first two laws of thermodynamics. In particular, the derived antisymmetric bracket enjoys the Jacobi identity. To obtain this formulation we suitably transform the equation and show that important physical quantities derived from the solution remain invariant under such a transformation.
\end{abstract}
\noindent{\it Keywords\/}: Radiative transfer, polarization, metriplectic formulation, structure-preservation
\maketitle

\section{Introduction}
We investigate the mathematical structure of the radiative transfer equation (RTE), which is a frequently used model in many applications, for example, atmospheric science \cite{Hansen1974, VanDiedenhoven2006} or biomedical imaging \cite{He2021,Zhu2022}.
While the RTE was first derived based on phenomenological arguments \cite{Chandrasekhar1960}, it was later given a firm basis as a high-frequency limit of the Maxwell equations \cite{Papanicolaou1975,Ryzhik1996}. 

The RTE is a linear integro-differential equation that describes the propagation of polarized light through a medium with varying refractive index and its scattering due to random fluctuations. 
These two processes, i.e., transport and scattering, often require a different treatment; for instance, in the mathematical analysis of the RTE, see, e.g., \cite{DautrayLions6,Bosboom2022}, or in its numerical approximation by discrete ordinates methods \cite{LewisMiller84}.

As we will show in this paper, the transport process can be described by a Hamiltonian flow, while the scattering process constitutes a dissipative process.
We then combine these two processes into one theoretical framework by using the so-called metriplectic formalism, originally conceived by Morrison \cite{Morrison1984,Morrison1984b,Morrison1986}.
The metriplectic formalism not only illuminates both the conservation and dissipation properties of an equation and gives a direct link to thermodynamics, but it also facilitates the construction of structure-preserving discretizations \cite{Suzuki2016, Kraus2017, Hirvijoki2021}.
When unpolarized light is considered the RTE reduces to a single equation for the light intensity \cite{Ryzhik1996}, and a metriplectic formulation for this case has recently been reported in \cite{Schlottbom2021}. The extension to polarized radiation considered here requires novel arguments.

The outline of this paper is as follows: In Section \ref{sec: Metriplectic_Formalism} we give a brief overview of the metriplectic formalism and explain that it can be used to reproduce the first two laws of thermodynamics. 
In Section \ref{sec: RTE} we introduce the radiative transfer equation, relate the solution to the Stokes parameters describing polarized light, and derive our main results on the metriplectic formulation for the radiative transfer equation. 
Furthermore, we show that for unpolarized light this metriplectic formulation reduces to one reported in \cite{Schlottbom2021}. 
In Section \ref{sec: Coupling} we cast the RTE in a suitable form by eliminating the term that describes optical rotation, which requires a proper relation of the RTE to Maxwell's equations. Additionally, we show that this procedure does not influence the physical quantities described by the solution. 
In Section \ref{sec: Conc-Disc} we give a conclusion on our work and discuss some open questions that could be of interest for future research. Lastly, in the appendix we give some of the tedious proofs of propositions in the previous sections.

\section{Introduction to the metriplectic formalism}
\label{sec: Metriplectic_Formalism}
The metriplectic formalism \cite{Kaufman1982, Kaufman1984, Morrison1984, Morrison1984b, Morrison1986, Grmela1984, Grmela1985} provides a framework for the description of systems that contain both Hamiltonian and dissipative parts. The Hamiltonian part of the evolution of such a system is determined by a Poisson bracket $\pbra{\cdot}{\cdot}$ and the Hamiltonian functional $\fh$, which is usually related to the total energy of the system and which stays constant in time. The dissipative part of the evolution is determined by a metric bracket $\mbra{\cdot}{\cdot}$ and an entropy functional $\fs$ that evolves monotonically in time. The purpose of this section is to give an abstract outline of this formalism.
Next, we describe the metriplectic framework in more detail. The presentation follows \cite{Morrison1986}.

\subsection{General framework}
\label{sec:general}

Let $H$ be a Hilbert space of a certain class of functions on a domain $\dom$, with an inner product $\langle\cdot,\cdot\rangle$, and let $u\in H$ be a dynamical variable. In the metriplectic formalism the time evolution of any functional $\ff$ of $u$ is given by
\begin{align}
\dfrac{\dd\ff}{\dd t} = \pbra{\ff}{\fg} + \mbra{\ff}{\fg},
\end{align}
where $\fg = \fh - \fs$ is called the free energy functional, analogous to the Gibbs free energy from thermodynamics, and $\fh$ and $\fs$ are the Hamiltonian and entropy functionals, respectively.
The Poisson bracket, which describes the ideal conservative evolution of the system, is real bilinear, antisymmetric, and satisfies Leibniz's rule and the Jacobi identity:
\begin{align}
\pbra{a\fa+b\fb}{\fc} &= a\pbra{\fa}{\fc}+b\pbra{\fb}{\fc},\\
\pbra{\fa}{\fb} &= -\pbra{\fb}{\fa},\\
\label{eq:leibniz_rule}
\pbra{\fa\fb}{\fc} &= \fa\pbra{\fb}{\fc}+\pbra{\fa}{\fc}\fb,\\
\label{eq:jacobi_identity}
  \pbra{ \pbra{ \fa }{ \fb } }{ \fc } &+ \pbra{ \pbra{ \fb }{ \fc } }{ \fa } +\pbra{ \pbra{ \fc }{ \fa } }{ \fb }= 0,
\end{align}
for arbitrary functionals $\fa,\fb,\fc$ of $u$ and real numbers $a, b$. 
The Poisson bracket can be represented via an anti-self adjoint operator $\fj(u):H\to H$ by
\begin{align}
    \pbra{\fa}{\fb} = \left\langle\dfrac{\de \fa}{\de u},\fj(u)\dfrac{\de \fb}{\de u} \right\rangle,
\end{align}
where $\dfrac{\de \fa}{\de u} \in H$ denotes the functional derivative of $\fa$, i.e., for any $v\in H$,
\begin{align}\label{eq:func_der}
    \left\langle \dfrac{\de \fa}{\de u} ,v\right\rangle = \lim_{\varepsilon\to 0}\frac{1}{\varepsilon}\left(\fa[u+\varepsilon v]-\fa[u]\right).
\end{align}
Furthermore, the entropy functional $\fs$ is a Casimir invariant of the Poisson bracket, meaning that for all functionals $\fa$ it holds that
\begin{align}
    \pbra{\fa}{\fs} = 0,
\end{align}
or equivalently $\fj(u)\dfrac{\de \fs}{\de u}=0$.
The metric bracket describes the dissipative effects of the system.
It is real bilinear, symmetric and negative semidefinite, i.e.,
\begin{align}
    \mbra{a\fa+b\fb}{\fc} &= a\mbra{\fa}{\fc}+b\mbra{\fb}{\fc},\\
    \mbra{\fa}{\fb} &= \mbra{\fb}{\fa},\\
    \mbra{\fa}{\fa} &\leq 0,
\end{align}
for arbitrary functionals $\fa,\fb,$ and $\fc$ of $u$ and real numbers $a$ and $b$. 
The metric bracket $\mbra{\cdot}{\cdot}$ has the representation
\begin{align}
    \mbra{\fa}{\fb} = \left\langle \dfrac{\de \fa}{\de u} , \fk(u)\, \dfrac{\de \fb}{\de u} \right\rangle,
\end{align}
where $\fk(u):H\to H$ is a self-adjoint operator.
Additionally, the Hamiltonian functional $\fh$ has zero metric bracket with any other functional of $u$, or equivalently $\fk(u)\dfrac{\de \fh}{\de u}=0$.
From the above-mentioned properties of these brackets it then follows that
\begin{align}
\dfrac{\dd\fh}{\dd t} &= \pbra{\fh}{\fg} + \mbra{\fh}{\fg} = \pbra{\fh}{\fh} = 0 , \\
\dfrac{\dd\fs}{\dd t} &= \pbra{\fs}{\fg} + \mbra{\fs}{\fg} = -\mbra{\fs}{\fs} \geq 0 ,
\end{align}
reproducing the First and Second Law of Thermodynamics.
\section{The radiative transfer equation}
\label{sec: RTE}
The radiative transfer equation can be used to model the propagation of light through a nonhomogeneous medium with many scatterers. We start by introducing the radiative transfer equation in its form frequently encountered in the literature and illustrate the physical meaning of the solution. We then present a suitable Lie bracket in Section~\ref{sec:Lie}. After exploiting the freedom of choosing a basis for the plane perpendicular to the transport direction, we construct the ingredients for the metriplectic formalism in Section~\ref{sec:metriplectic}, which constitute the main results of this paper. Section~\ref{sec:scalar} discusses the relation of the metriplectic formulation for the case of unpolarized light to the metriplectic formulation for the scalar radiative transfer equation derived in \cite{Schlottbom2021}. The proofs and derivations of the results that follow are deferred to later sections and the appendix.

\subsection{Radiative transfer}
The radiative transfer equation is a high-frequency limit of the Maxwell equations describing the propagation of polarized light. It is given by \cite{Chandrasekhar1960,Ryzhik1996}
\begin{align}
    \label{eq: RTE}
    \frac{\partial \MW}{\partial t} + \nabla_k\omega\cdot\nabla_x \MW - \nabla_x \omega\cdot\nabla_k \MW + \Sigma \MW=\MN\MW-\MW\MN + S(\MW).
\end{align}
The solution to this equation $\MW = \MW(x,k,t)$ is called the coherence matrix and is a $2\times2$ matrix-valued function, depending on position $x\in \RR^3$, wave vector $k\in \RR^3\backslash\{0\}$ and time $t\geq 0$.
The second and third terms on the left-hand side are transport terms describing the refraction of the light as it propagates through an inhomogeneous medium. The expressions of the quantities in this equation follow from a proper derivation of the RTE from the Maxwell equations \cite{Ryzhik1996}. The function $\omega(x,k)=v(x)|k|$ is the dispersion relation, with $v(x)=1/\sqrt{\epsilon(x)\mu(x)}$ the velocity function, and $\epsilon$ and $\mu$ the scalar local permittivity and permeability, respectively. 
The optical rotation matrix $\MN$ is given by
\begin{align}
    \MN(x,k) = 
    \begin{pmatrix}
    0 & n(x,k)\\
    -n(x,k) & 0
    \end{pmatrix}=
    n(x,k)\MJ,
\end{align}
where  $n(x,k)$ is a scalar function related to the rate of change in linear polarization, and $\MJ$ is the canonical symplectic matrix
\begin{align*}
    \MJ = \begin{pmatrix}
    0& 1\\
    -1& 0
    \end{pmatrix}.
\end{align*}
The coupling term describes the rotation of the polarization plane due to the torsion of light rays in an inhomogeneous medium, also known as Rytov's law \cite{Vinitskii1990}. The scattering operator $S$, given by
\begin{align}
    \label{eq: Scattering}
    (S\MW)(x,k,t) = \int_{|k'| = |k|} \sigma(x,k, k') \MT(x,k,k') \MW(x,k',t)\MT(x,k,k')^* \,\rmd\lambda(k'),
\end{align}
describes the redistribution of light over different propagation directions and polarizations. Here $\sigma$ is a positive function that is symmetric in $k$ and $k'$, $\lambda$ is the surface measure, and $\MT(x,k,k')$ is a real $2\times 2$ matrix-valued function that satisfies $\MT(x,k,k') = \MT(x,k',k)^*$, where the superscript $^*$ denotes transposition and complex conjugation.
The total scattering rate is described by the scalar function $\Sigma(x,k)$, which is given by
\begin{align}
    \label{eq: ScatteringRate}
    \Sigma(x,k)I_2 = \int_{|k'| = |k|} \sigma(x,k, k') \MT(x,k,k') \MT(x,k,k')^*\,\rmd\lambda(k'),
\end{align}
with $I_2$ denoting the $2\times 2$ identity matrix. It has been shown that under mild assumptions on $v(x), n(x,k)$ and $\Sigma(x,k)$ equation \eqref{eq: RTE} is well-posed and that the solution $\MW(x,k,t)$ is Hermitian \cite{Bosboom2022}. 
Furthermore, the solution is related to the four Stokes parameters $I,Q,U,V$, see \cite{Chandrasekhar1960,Ryzhik1996}, as
\begin{align}\label{eq: Stokes}
    \MW = \frac{1}{2}
    \begin{pmatrix}
    I+Q & U+\rmi V\\
    U-\rmi V & I-Q
    \end{pmatrix}.
\end{align}
Here $I$ denotes the intensity of the light, $Q$ and $U$ describe different modes of linear polarization, and $V$ is related to the amount of circular polarization. 

There are two major challenges for forming a metriplectic formulation of the radiative transfer equation as stated in \eqref{eq: RTE}. First, the advective term $\nabla_k\omega\cdot\nabla_x \MW-\nabla_x\omega\cdot\nabla_k \MW$ combines the scalar-valued function $\omega$ with the matrix-valued solution $\MW$, while the metriplectic formalism requires a pairing, i.e., a bracket, for two matrix-valued functions.
Second, the optical rotation term $\MN$ makes it difficult  to construct a Poisson bracket, i.e., a bracket that enjoys the Jacobi identity. In the next two subsections we address these challenges.

\subsection{A Lie bracket for matrix-valued functions}\label{sec:Lie}
Denote by $\sigma_i$ for $i=1,2,3$ the normalized Pauli-matrices
\begin{align*}
    \sigma_1 = \frac{1}{\sqrt{2}}\begin{pmatrix}
        0 & 1 \\
        1 & 0
    \end{pmatrix}, 
    \sigma_2 =  \frac{1}{\sqrt{2}}\begin{pmatrix}
        0 & -i\\
        i & 0
    \end{pmatrix},
    \sigma_3 =  \frac{1}{\sqrt{2}}\begin{pmatrix}
        1 & 0\\ 0 & -1
    \end{pmatrix}.
\end{align*}
Together with the identity matrix $\sigma_0=I_2/\sqrt{2}$ these matrices form a basis for the real vector space of $2\times 2$ Hermitian matrices.
 It holds that
\begin{align*}
    \trace\left({\sigma_j}{\sigma_k}\right)=\delta_{j,k},\quad\text{for } j,k\in\{0,1,2,3\}.
\end{align*}
Hence, any $2\times2$ Hermitian matrix can be decomposed as
\begin{align*}
    \MW = \sum_{j=0}^3 W_j \sigma_j,\quad\text{with } W_j = \trace\left(\MW \sigma_j\right)\in\RR.
\end{align*}
For sufficiently smooth $2\times 2$ Hermitian matrix-valued functions $\MU,\MV$ we define the bracket
\begin{align}
    \label{eq: Matrix_Bracket}
    \left[\MU,\MV\right]_{M} &=\sum_{j=0}^3 [U_j,V_j]\sigma_j,
\end{align}
where $\cbra{\cdot}{\cdot}$ is the canonical Poisson bracket given by
\begin{align}\label{eq: cbra}
    \cbra{f}{g} = \nabla_xf\cdot\nabla_kg-\nabla_kf\cdot\nabla_x g
\end{align}
for sufficiently smooth functions $f,g$. When 
$\MU=\Omega=\sum_{j=0}^3\omega\sigma_j$ and $\MV=\MW$, this bracket gives the transport term in \eqref{eq: RTE}, i.e., 
\begin{align*}
\cbra{\Omega}{\MW}_M = \sum_{j=0}^3\cbra{\omega}{W_j}\sigma_j = \nabla_x\omega\cdot\nabla_k \MW - \nabla_k \omega\cdot\nabla_x \MW,
\end{align*}
where we use the subscript $M$ to emphasize that the bracket acts on matrix-valued functions. Furthermore, the following properties hold, making this bracket a Lie bracket.
\begin{proposition}
    \label{prop: Matrix_Bracket}
    For sufficiently smooth $2\times 2$ matrix-valued functions $\MU,\MV,\MW$, the bracket defined in \eqref{eq: Matrix_Bracket} is bilinear, anti-symmetric and satisfies the Jacobi identity
    \begin{align}
        \cbra{\cbra{\MU}{\MV}_M}{\MW}_M+\cbra{\cbra{\MV}{\MW}_M}{\MU}_M+\cbra{\cbra{\MW}{\MU}_M}{\MV}_M = 0.
    \end{align}
\end{proposition}
\begin{proof}
\changecolor{black}\hyperlink{Proof_Matrix_Bracket}{See the proof in the appendix.}\changecolor{blue}
\end{proof}

\subsection{Removal of the optical rotation term \texorpdfstring{$\MN$}{}}
From the derivation of the radiative transfer equation in \cite{Ryzhik1996} it can be seen that the optical rotation term depends on a choice of an orthonormal basis for the plane perpendicular to the vector $k$. 
There is freedom in choosing this basis and different choices give different values of the optical rotation term. Furthermore, there exists a specific choice of basis that eliminates the optical rotation term completely \cite{Ryzhik1996, Lewis1966}. 
We recall the procedure of finding this basis in Section \ref{sec: Coupling} and show that this choice does not influence the physical relevance of the solution $\MW$, which is contained in certain expressions in terms of the Stokes parameters.
Using this basis, and the bracket introduced in \eqref{eq: Matrix_Bracket}, the radiative transfer equation is given by
\begin{align}
    \label{eq: RTENoCouplingBracket}
    \frac{\partial \MW}{\partial t} = \cbra{\Omega}{\MW}_M +S(\MW)- \Sigma \MW,
\end{align}
which is the form we employ for constructing our metriplectic formulation in the following.
Let us note that the matrix $\MT(x,k,k')$ depends on the choice of basis as well, but that the structural properties of the corresponding $S(\MW)$ and $\Sigma$ used in the metriplectic formulation in \eqref{eq: RTENoCouplingBracket} are still of the form given in \eqref{eq: Scattering} and \eqref{eq: ScatteringRate}, independent of the choice of this basis, see Remark~\ref{remark: scattering} below.

\subsection{Constructing the metriplectic formulation}\label{sec:metriplectic}
We consider the radiative transfer equation on the domain $(x,k)\in \dom = \RR^3\times \RR^3\backslash\{0\}$. Let $\HH$ denote the space of $2\times2$ Hermitian matrices and let $L^2(\dom,\HH)$ be the Lebesgue space of square-integrable functions from $\dom$ to $\HH$, equipped with the inner product
\begin{align}
    \label{eq: InnerProductMatrix}
    \langle\MU,\MV\rangle = \int_{\dom}\Tr\left(\MU\MV\right)\dx\dk.
\end{align}
This is the natural space to consider for the metriplectic formulation since, for sufficiently smooth initial data, equation \eqref{eq: RTENoCouplingBracket} has a unique solution $\MW \in C^1([0,T],L^2(\dom,\HH))$ such that $\cbra{\Omega}{\MW}_M\in C^0([0,T];L^2(\dom,\HH))$ \cite{Bosboom2022}. 

By multiplying equation \eqref{eq: RTENoCouplingBracket} from the left with $\MU\in C^\infty_0(\dom,\HH)$, taking the trace and integrating over $\dom$ we obtain the equation
\begin{align*}
    \left\langle\MU,\fracp{\MW}{t}\right\rangle &= \Big\langle\MU,\cbra{\Omega}{\MW}_M\Big\rangle-\Big\langle\MU(x,k),\Sigma(x,k)\MW(x,k)\Big\rangle\\
    &+\left\langle\MU(x,k),\int_{|k|=|k'|}\sigma(x,k, k')\MT(x,k,k')\MW(x,k',t)\MT(x,k',k)\dk'\right\rangle.
\end{align*}
By applying integration by parts to the term $\langle\MU,\cbra{\Omega}{\MW}_M\rangle$ and observing that boundary terms vanish due to the compact support of $\MU$, we deduce that
\begin{multline}
    \label{eq: WeakFormRTE}
    \frac{d}{dt}\Big\langle\MU,\MW\Big\rangle = \Big\langle\MW,\cbra{\MU}{\Omega}_M\Big\rangle-\Big\langle\MU(x,k),\Sigma(x,k)\MW(x,k)\Big\rangle\\
    +\left\langle\MU(x,k),\int_{|k|=|k'|}\sigma(x,k, k')\MT(x,k,k')\MW(x,k',t)\MT(x,k',k)\dk'\right\rangle.
\end{multline}
The first term on the right-hand side of this equation is antisymmetric in $\MU$ and $\Omega$ and it is thus natural to define the bracket
\begin{align}
    \label{eq: Symp_Bracket_RTE}
    \pbra{\fa}{\fb} = \int_\dom\Tr\left(\MW\cbra{\fracd{\fa}{\MW}}{\fracd{\fb}{\MW}}_M\right)\dx\dk.
\end{align}
Here $\fracd{\fa}{\MW}$ is the functional derivative as defined in \eqref{eq:func_der}.
By applying the Fubini theorem to the last two terms on the right-hand side of equation \eqref{eq: WeakFormRTE} it can be seen that they are symmetric in $\MU$ and $\MW$ and thus it is natural to define the bracket that describes these terms by
\begin{multline}
    \label{eq: Metr_Bracket_RTE}
    \mbra{\fa}{\fb} = \frac{1}{2}\int_\dom\int_{|k|=|k'|}\sigma(x,k,k')\Tr\Bigg[\left(\MT(x,k,k')\fracd{\fa}{\MW'}-\fracd{\fa}{\MW}\MT(x,k,k')\right)\\
    \cdot\left(\MT(x,k',k)\fracd{\fb}{\MW}-\fracd{\fb}{\MW'}\MT(x,k',k)\right)\Bigg]\dk'\dk\dx.
\end{multline}
Here $\fracd{\fa}{\MW'}$ denotes the evaluation of $\fracd{\fa}{\MW}$ in $(x,k')$.
Using these brackets, equation \eqref{eq: WeakFormRTE} takes the form
\begin{align}
    \label{eq: Brackets_Separate}
    \frac{d\ff_\MU}{dt} = \pbra{\ff_\MU}{\fh}-\mbra{\ff_\MU}{\fs},
\end{align}
with the functionals
\begin{align}
    \label{eq: Test_Functional_RTE}
    \ff_\MU[\MW] &= \int_\dom\Tr\left(\MU\MW\right)\dx\dk,\\
    \label{eq: Hamiltonian_Functional_RTE}
    \fh[\MW] &= \int_\dom \Tr\left(\Omega\MW\right)\dx\dk,\\
    \label{eq: Entropy_Functional_RTE}
    \fs[\MW] &= -\frac{1}{2}\int_\dom\Tr\left(\MW^2\right)\dx\dk.
\end{align}
In the above, the functional derivatives are taken with respect to the inner product defined in \eqref{eq: InnerProductMatrix}. It can be shown that the bracket in \eqref{eq: Symp_Bracket_RTE} is a Poisson bracket:
\begin{proposition}
    \label{prop: Symp_Bracket}
    For sufficiently smooth functionals $\fa, \fb,\fc$ of $\MW$, the bracket in \eqref{eq: Symp_Bracket_RTE} is bilinear, antisymmetric. Additionally, it satisfies Leibniz's rule
    \begin{align*}
        \pbra{\fa\fb}{\fc} = \fa\pbra{\fb}{\fc}+\pbra{\fa}{\fc}\fb,
    \end{align*}
    and the Jacobi identity
    \begin{align*}
        \pbra{\pbra{\fa}{\fb}}{\fc}+\pbra{\pbra{\fb}{\fc}}{\fa}+\pbra{\pbra{\fc}{\fa}}{\fb}=0.
    \end{align*}
\end{proposition}

\begin{proof}
    Bilinearity and antisymmetry of this bracket follow directly from the properties of the bracket in \eqref{eq: Matrix_Bracket}. Furthermore, the Leibniz rule follows from the product rule for functional derivatives
    \begin{align*}
        \fracd{\fa\fb}{\MW} = \fa\fracd{\fb}{\MW}+\fracd{\fa}{\MW}\fb.
    \end{align*}
    The Jacobi identity follows from the Jacobi identity for $\cbra{\cdot}{\cdot}_M$ and the procedure given in \cite[p. 612]{Abraham1988}. \changecolor{black}\hyperlink{Proof_Symp_Jacobi}{For self-containedness we prove the Jacobi identity in the appendix.}\changecolor{blue}
\end{proof}
Similarly it can be shown that \eqref{eq: Metr_Bracket_RTE} defines a metric bracket:
\begin{proposition}
    The bracket in \eqref{eq: Metr_Bracket_RTE} is bilinear, symmetric, and negative semidefinite.
\end{proposition}
\begin{proof}
    The bracket in \eqref{eq: Metr_Bracket_RTE} is clearly bilinear. Symmetry follows by applying the Fubini theorem to the $k$ and $k'$ integrals and relabeling the variables $k$ and  $k'$. Furthermore, negative-semidefiniteness of the bracket has been shown in \cite[Proposition 4.5]{Bosboom2022}.
\end{proof}
The final step in the construction of the metriplectic formulation of the RTE is to verify that $\fh$ and $\fs$ defined in \eqref{eq: Hamiltonian_Functional_RTE} and \eqref{eq: Entropy_Functional_RTE}, respectively, are Casimirs of the corresponding brackets, which is established next.
\begin{proposition}
    \label{prop: Casimirs}
    The entropy functional $\fs$ is a Casimir of the Poisson bracket and the Hamiltonian functional $\fh$ is a Casimir of the metric bracket.
\end{proposition}
\begin{proof}
    \changecolor{black}\hyperlink{Proof_Casimirs}{The proof can be found in the appendix.}\changecolor{blue}
\end{proof}
Employing Proposition~\ref{prop: Casimirs}, equation \eqref{eq: Brackets_Separate} can be written as
\begin{align}
    \label{eq: RTE_Metriplectic_Formulation}
    \frac{d\ff_\MU}{dt} = \pbra{\ff_\MU}{\fg}+\mbra{\ff_\MU}{\fg}
\end{align}
with the free energy functional
\begin{align}
    \label{eq: Free_Energy_Functional_RTE}
    \fg[\MW] = \fh[\MW]-\fs[\MW].
\end{align}
These brackets provide a proper metriplectic formulation of the radiative transfer equation with conservation of the Hamiltonian and dissipation of the entropy following as in Section~\ref{sec: Metriplectic_Formalism}.
\begin{remark}
    We note that in the vacuum case, i.e., $\sigma=0$ and $v$ constant, the metric bracket in \eqref{eq: Metr_Bracket_RTE} vanishes, and equation \eqref{eq: RTE_Metriplectic_Formulation} contains only Hamiltonian dynamics.\\
    Furthermore, we can consider the diffusion limit, which is valid for propagation distances that are large compared to the mean free path $v(x)/\Sigma$. In this limit the trace of $\MW(x,k,t)$ converges, for constant $v$, to a scalar function $\phi(x,r,t)$ satisfying the diffusion equation \cite[(5.29)]{Ryzhik1996}
    \begin{align*}
        \fracp{\phi}{t} = \nabla_x\cdot\left(D\nabla_x\phi\right),
    \end{align*}
    with diffusion coefficient $D(x,r)>0$. This can be described as a purely dissipative system
    \begin{align*}
        \fracp{\ff}{t} = \mbra{\ff}{\fg_D}_D,
    \end{align*}
    with metric bracket
    \begin{align*}
        \mbra{\fa}{\fb}_D = -\int_{\RR^3}\int_0^\infty D\nabla_x\fracd{\fa}{\phi}\cdot\nabla_x\fracd{\fb}{\phi}\dx\, r^2\dr,
    \end{align*}
    and a free energy functional $\fg_D = \fh_D-\fs_D$, with
    \begin{align*}
        &\fh_D[\phi] = \int_{\RR^3}\int_0^\infty\phi\dx\, r^2\dr, \qquad\fs_D[\phi] = -\frac{1}{2}\int_{\RR^3}\int_0^\infty\phi^2\dx\, r^2\dr.
    \end{align*}
\end{remark}
\subsection{Scalar radiative transfer}\label{sec:scalar}
If we consider only unpolarized light, the Stokes parameters $Q,U$ and $V$ are identically zero and the coherence matrix can be written as
\begin{align*}
    \MW(x,k,t) = \frac{1}{2}\begin{pmatrix}
    I(x,k,t) & 0\\
    0 & I(x,k,t)
    \end{pmatrix}.
\end{align*}
In this case equation \eqref{eq: RTENoCouplingBracket} reduces to the scalar radiative transfer equation for the intensity $I$ by taking traces, i.e.,
\begin{align}
    \label{eq: RTE_scalar}
    \fracp{I}{t}+\cbra{I}{\omega} = \int_{|k|=|k'}\sigma'(x,k,k')I(x,k',t)\dk' -\Sigma'(x,k)I(x,k,t),
\end{align}
with
\begin{align*}
    \sigma'(x,k,k') &= \frac{1}{2}\sigma(x,k,k')\Tr\left[\MT(x,k,k')\MT(x,k',k)\right],\\
    \Sigma'(x,k) &= \int_{|k|=|k'|}\sigma'(x,k,k')\dk'.
\end{align*}
In this case the functionals \eqref{eq: Hamiltonian_Functional_RTE}, \eqref{eq: Entropy_Functional_RTE} and \eqref{eq: Free_Energy_Functional_RTE} reduce to
\begin{align}
    \fh[\MW] &= \int_\dom\Tr\left(\Omega\MW\right)\dx\dk = \int_\dom\omega(x,k)I(x,k,t)\dx\dk = \tilde\fh[I],\\
    \fs[\MW] &= -\frac{1}{2}\int_\dom\Tr\left(\MW^2\right)\dx\dk = -\frac{1}{4}\int_\dom I^2\dx\dk = \tilde\fs[I],\\
    \tilde\fg[I] &= \tilde\fh[I]-\tilde\fs[I].
\end{align}
Similarly, setting $\MU=\varphi I_2$ for $\varphi\in C^\infty_0(\dom)$,
the functional $\ff_\MU$ defined in equation \eqref{eq: Test_Functional_RTE} reduces to
\begin{align}
    \ff_\MU[\MW] = \int_\dom\Tr\left(\MU\MW\right)\dx\dk = \int_\dom \varphi(x,k)I(x,k,t)\dx\dk = \tilde\ff_\varphi[I].
\end{align}
Here, $\fracd{\tilde\ff[I]}{I}$ denotes the functional derivative w.r.t.\@ the usual $L^2$-inner product on $\dom$. In the case of unpolarized light the brackets in \eqref{eq: Symp_Bracket_RTE} and \eqref{eq: Metr_Bracket_RTE} reduce to
\begin{align}
    &\pbra{\fa}{\fb}_s = \int_\dom I\cbra{\fracd{\fa}{I}}{\fracd{\fb}{I}}\dx\dk,\\
    &\mbra{\fa}{\fb}_s = \int_\dom \sigma'(x,k,k')\left(\fracd{\fa}{I}-\fracd{\fa}{I'}\right)\left(\fracd{\fb}{I'}-\fracd{\fb}{I}\right)\dk'\dk\dx,
\end{align}
and the metriplectic formulation of \eqref{eq: RTE_scalar} becomes
\begin{align*}
    \frac{d\tilde\ff_\varphi}{dt} = \pbra{\tilde\ff_\varphi}{\tilde\fg}_s +\mbra{\tilde\ff_\varphi}{\tilde\fg}_s,
\end{align*}
which is in full agreement with the formulation derived in \cite{Schlottbom2021}.

\section{Optical rotation}
\label{sec: Coupling}

This section is concerned with the procedure of eliminating the optical rotation term in equation \eqref{eq: RTE}. For the convenience of the reader we first recall the key aspects of the  derivation of the radiative transfer equation from Maxwell's equations and discuss the physical meaning of the solution.

\subsection{Derivation of the radiative transfer equation from Maxwell's equation}
Following \cite{Ryzhik1996}, equation \eqref{eq: RTE} can be derived from a transformation of the Maxwell equations for the electric field $E$ and the magnetic field $H$:
\begin{align}
    \begin{pmatrix}
        \epsilon & 0\\
        0 & \mu
    \end{pmatrix}
    \fracp{}{t}
    \begin{pmatrix}
        E\\
        H
    \end{pmatrix}+
    \begin{pmatrix}
        0&-\nabla\times\\
        \nabla\times & 0
    \end{pmatrix}
    \begin{pmatrix}
        E\\
        H
    \end{pmatrix}=0.
\end{align}
Using the vector $u=\begin{pmatrix} E\\H
\end{pmatrix}$ this equation can be cast into the standard form of a symmetric first-order hyperbolic system
\begin{align}
    A(x)\fracp{u}{t}+\sum_{i=1}^3D^i\fracp{u}{x_i}=0.
\end{align}
In the procedure of deriving the RTE, the limiting Wigner matrix $\MW^{(0)}$ is defined as a high-frequency Wigner transformation of $u$:
\begin{align}
    \label{eq: limiting_wigner_matrix}
    \MW^{(0)}(x,k,t) = \lim_{\varepsilon\to 0}\left(\frac{1}{2\pi}\right)^3\int_{\RR^3}\rme^{\rmi k\cdot y}u(t,x-\frac{1}{2}\varepsilon y)u(t,x+\frac{1}{2}\varepsilon y)^*\, \rmd y.
\end{align}
Since the electric and magnetic fields corresponding to a light wave always lie in the plane perpendicular to the direction of propagation $\hat{k}= {k}/{|k|}$. They are respectively normal and binormal vectors to $\hat{k}$. By choosing a pair of orthonormal vectors $\left\{z^{(1)}(x,k), z^{(2)}(x,k)\right\}$ that span the plane perpendicular to $\hat{k}$, the limiting Wigner matrix can be expanded as
\begin{align*}
    \MW^{(0)}(x,k,t) = \sum_{i,j=1}^2a_{ij}^+B^{i,j}_++a_{ij}^-B^{i,j}_- = \sum_{i,j=1}^2a_{ij}^+b^{(i)}_+b^{(j)^T}_++a_{ij}^-b^{(i)}_-b^{(j)^T}_-,
\end{align*}
with
\begin{align*}
    &b^{(1)}_+(x,k) = \left(\sqrt{\frac{1}{2\epsilon}}z^{(1)}, \sqrt{\frac{1}{2\mu}}z^{(2)}\right), \qquad b^{(2)}_+(x,k) = \left(\sqrt{\frac{1}{2\epsilon}}z^{(2)}, -\sqrt{\frac{1}{2\mu}}z^{(1)}\right),\\
    &b^{(1)}_-(x,k) = \left(\sqrt{\frac{1}{2\epsilon}}z^{(1)}, -\sqrt{\frac{1}{2\mu}}z^{(2)}\right), \qquad b^{(2)}_-(x,k) = \left(\sqrt{\frac{1}{2\epsilon}}z^{(2)}, \sqrt{\frac{1}{2\mu}}z^{(1)}\right).
\end{align*}
Here the $+$ and $-$ indicate right and left propagation modes, respectively. Furthermore, the matrices $B_\pm^{i,j}$ are orthonormal in the following sense
\begin{align*}
    \Tr\left(AB^{i,j}_sAB^{k,l}_t\right) = \delta_{i,k}\delta_{j,l}\delta_{s,t}.
\end{align*}
Therefore, the coefficients $a^\pm_{i,j}(x,k,t)$ can be determined as
\begin{align*}
    a^\pm_{ij} = \Tr\left(A\MW^{(0)*}AB^{ij}_\pm\right).
\end{align*}
By symmetry (see \cite{Ryzhik1996} for the full details) the left propagating modes can be derived from the right propagation modes, so we are only interested in the part
\begin{align*}
    \MW^{(0)}_+(x,k,t) = \sum_{i,j=1}^2a_{ij}^+B^{ij}_+.
\end{align*}
The expansion coefficients $a_{ij}^+(x,k,t)$ determine the coherence matrix:
\begin{align}\label{eq: def_coherence}
    \MW_{i,j}(x,k,t) = a_{ij}^+(x,k,t) .
\end{align}
The coherence matrix then satisfies \eqref{eq: RTE} with the optical rotation term given by
\begin{align}
    \label{eq: Coupling_General}
    \MN_{mn}(x,k) &= \left(b^{(n)}_+,D^i\frac{\partial b^{(m)}_+}{\partial x^i}\right) - \frac{\partial \omega}{\partial x^i}\left(A(x)b^{(n)}_+,\frac{\partial b^{(m)}_+}{\partial k_i}\right) -\frac{1}{2}\frac{\partial^2\omega}{\partial x^i\partial k_i}\delta_{nm},
\end{align}
where $\left(\cdot,\cdot\right)$ denotes the inner product
\begin{align*}
    \left(u,v\right) = \sum_{i=1}^6u_i(x,k)v_i(x,k).
\end{align*}
From this derivation we can conclude that both the coherence matrix $\MW$ and the optical rotation matrix $\MN$ depend on the choice of basis $\left\{z^{(1)}(x,k), z^{(2)}(x,k)\right\}$ of the plane perpendicular to $\hat{k}$. However, many quantities derived from $\MW$ are independent of this choice:

\begin{theorem}Let $I,Q,U,V$ be the Stokes parameter components of $\MW$ in the basis determined by the vectors $z^{(1)}(x,k)$ and $z^{(2)}(x,k)$ as defined in \eqref{eq: Stokes}. Then for any continuous function $f:\RR^3\to\RR$ it holds that $f(I,Q^2+U^2,V)$ is invariant under a rotation of $z^{(1)}(x,k)$ and $z^{(2)}(x,k)$ in the plane perpendicular to $\hat{k}$.
\end{theorem}

\begin{proof}
    It suffices to show that the quantities $I,V$ and $Q^2+U^2$ are independent under rotations of the basis. Under a rotation of angle $\theta(x,k)$ in the plane perpendicular to $\hat{k}$ the vectors $z^{(1)}(x,k)$ and $z^{(2)}(x,k)$ transform into  vectors $z^{(1)'}(x,k)$ and $z^{(2)'}(x,k)$ as
\begin{align*}
    &z^{(1)'}(x,k) = \cos(\theta) z^{(1)}+\sin(\theta) z^{(2)},\\
    &z^{(2)'}(x,k) = -\sin(\theta) z^{(1)}+\cos(\theta) z^{(2)}.
\end{align*}
To simplify the notation we write the indices as $(B^{11}_+,B^{12}_+,B^{21}_+,B^{22}_+) = (B^1_+,B^2_+,B^3_+,B^4_+)$ and $(a^+_{11},a^+_{12},a^+_{21},a^+_{22}) = (a^+_1,a^+_2,a^+_3,a^+_4)$. 
From the above transformation rule and the definition of the corresponding matrices $B^{i}_+$ and $B^{i'}_+$, we obtain the relation
\begin{align*}
    \begin{pmatrix}
    B^{1'}_+\\
    B^{2'}_+\\
    B^{3'}_+\\
    B^{4'}_+
    \end{pmatrix}=
    (R
    \otimes I_6)
    \begin{pmatrix}
    B^{1}_+\\
    B^{2}_+\\
    B^{3}_+\\
    B^{4}_+
    \end{pmatrix},
\end{align*}
with the matrix
\begin{align*}
    R=\begin{pmatrix}
    \cos^2\theta & \cos\theta\sin\theta & \cos\theta\sin\theta & \sin^2\theta\\
    -\sin\theta\cos\theta & \cos^2\theta & -\sin^2\theta & \sin\theta\cos\theta\\
    -\sin\theta\cos\theta & -\sin^2\theta & \cos^2\theta & \sin\theta\cos\theta\\
    \sin^2\theta & -\sin\theta\cos\theta & -\sin\theta\cos\theta & \cos^2\theta
    \end{pmatrix}.
\end{align*}
The limiting Wigner matrix, however, is by definition independent of the choice of the basis and can be expanded in the transformed basis as
\begin{align*}
    \MW^{(0)}(x,k,t) = \sum_{i}^4a_{i}^{+'}B^{i'}_{+}+a_{i}^{-'}B^{i'}_{-}.
\end{align*}
Then the coefficients $a_{i}^{+'}(x,k,t)$ are determined by
\begin{align*}
    a_{i}^{+'} = \Tr\left(A\MW^{(0)*}AB^{i'}_+\right) &= \Tr\left(A\MW^{(0)*}_+AB^{i'}_+\right) = \sum_{k,l=1}^4a_{k}^+R_{i,l}\Tr\left(AB^{k}_+A B^{l}_+\right)\\
    &=\sum_{k,l=1}^4a_{k}^+R_{i,l}\delta_{k,l} = \sum_{k=1}^4R_{i,k}a_k^+.
\end{align*}
The Stokes parameters after the transformation are defined as, cf. \eqref{eq: Stokes} and \eqref{eq: def_coherence},
\begin{align*}
    &I' = a_{1}^{+'}+a_{4}^{+'}, \qquad Q' = a_{1}^{+'}-a_{4}^{+'}\\
    &U' = a_{2}^{+'}+a_{3}^{+'}, \qquad V' = \frac{1}{\rmi}(a_{2}^{+'}-a_{3}^{+'}),
\end{align*}
the transformation law for the Stokes parameters can then be written as
\begin{align*}
    \begin{pmatrix}
    I'\\
    Q'\\
    U'\\
    V'
    \end{pmatrix}=
    \begin{pmatrix}
    1 & 0 & 0 &0\\
    0 & \cos(2\theta) & \sin(2\theta) & 0\\
    0 & -\sin(2\theta) & \cos(2\theta) & 0\\
    0 & 0 & 0 &1
    \end{pmatrix}
    \begin{pmatrix}
    I\\
    Q\\
    U\\
    V
    \end{pmatrix}.
\end{align*}
The theorem then follows immediately from this transformation law.
\end{proof}
\subsection{Eliminating the optical rotation term}
Now that we know how equation \eqref{eq: RTE} transforms under a change of basis $\left\{z^{(1)}(x,k), z^{(2)}(x,k)\right\}$, we will next construct a particular basis such that the optical rotation matrix $\MN$ defined in \eqref{eq: Coupling_General} vanishes.

It can be seen that the third term on the right side of equation \eqref{eq: Coupling_General} is a diagonal term. This term does not affect equation \eqref{eq: RTE} since $\MN$ enters there in a commutator. It was added in \cite{Ryzhik1996} for convenience because it cancels the first term on the right-hand side of \eqref{eq: Coupling_General} when the basis functions $z^{(1)}(x,k)$ and $z^{(2)}(x,k)$ are chosen as a pair of orthonormal vectors perpendicular to the vectors $\hat{k}$.
Hence, by expanding the remaining term in \eqref{eq: Coupling_General}, it can be seen that
\begin{align*}
    \MN_{ij}(x,k) = -\frac{1}{2}\sum_{l,m=1}^2(\delta_{jl}\delta_{im}+\epsilon_{jl}\epsilon_{im})z^{(l)}(x,k)\cdot\left(\nabla_xv(x)\cdot \nabla_k\right)z^{(m)}(x,k),
\end{align*}
with $\epsilon_{ij}$ the Levi-Civita symbol. Since $z^{(1)}(x,k)$ and $z^{(2)}(x,k)$ must be perpendicular to $\hat{k}$, the rotation term vanishes if both $\fracp{z^{(1)}}{s}$ and $\fracp{z^{(2)}}{s}$ are parallel to $\hat{k}$, with the material derivative $\fracp{}{s}$ given by
\begin{align*}
    \fracp{}{s} := (\nabla_xv(x)\cdot\nabla_k).
\end{align*}
To construct these vectors we follow the procedure in \cite{Lewis1966}. We first construct a Frenet-Serret frame $(T,N,B)$, with the tangent vector $T=\hat{k}$ and the normal vector $N$ and binormal vector $B$ defined by
\begin{align*}
    &N := \frac{1}{\kappa}\frac{\partial T}{\partial s}, \qquad B := T\times N,
\end{align*}
which satisfy the equations
\begin{align*}
    &\frac{\partial N}{\partial s} = -\kappa T+\tau B, \qquad\frac{\partial B}{\partial s} = -\tau N,
\end{align*}
with $\kappa = |\frac{\partial T}{\partial s}|$ the curvature and $\tau$ the torsion of the curve. The vectors $N$ and $B$ are themselves not the correct choice of basis, but from the Frenet-Serret frame $(T,N,B)$ we can construct a Darboux frame $(T,P,Q)$ as
\begin{align*}
    &P := \cos\alpha N+\sin\alpha B,\\
    &Q := -\sin\alpha N+\cos\alpha B,
\end{align*}
and by choosing $\alpha$ such that $\frac{\partial \alpha}{\partial s} = -\tau$ it holds that
\begin{align*}
    &\frac{\partial P}{\partial s} = -\kappa\cos\alpha T,\\
    &\frac{\partial Q}{\partial s} = \kappa\sin\alpha T.
\end{align*}
Since $(T,P,Q)$ forms an orthonormal triplet and $\frac{\partial P}{\partial s}$ and $\frac{\partial Q}{\partial s}$ are parallel to $\hat{k}$, we can choose $z^{(1)}(x,k)=P$ and $z^{(2)}(x,k) = Q$ to make the optical rotation term vanish and obtain equation \eqref{eq: RTENoCouplingBracket}.

\begin{remark}
    \label{remark: scattering}
    Following the derivation of the RTE in \cite{Ryzhik1996}, the scattering matrix $\MT$ is defined as $\MT_{ij}(x,k,k') = z^{(i)}(x,k)\cdot z^{(j)}(x,k')$. From this expression, it can be seen that after any rotation of the basis vectors, the rotated scattering matrix $\tilde{\MT}$ still satisfies $\tilde{\MT}(x,k,k') = \tilde{\MT}(x,k',k)^*$, and the corresponding RTE can be written as in \eqref{eq: RTENoCouplingBracket}, with $S(\MW)$ and $\Sigma$ satisfying equations \eqref{eq: Scattering} and \eqref{eq: ScatteringRate}, respectively, with the rotated scattering matrix.
\end{remark}
\section{Conclusion and Discussion}
\label{sec: Conc-Disc}
In this paper, we constructed a metriplectic formulation of the radiative transfer equation with polarization by constructing a suitable matrix extension of the canonical Poisson bracket and by defining a suitably rotating basis for the polarization plane that eliminates the optical rotation term in the equations. 
In particular, this choice of basis allowed us to show that the antisymmetric bracket defined in \eqref{eq: Symp_Bracket_RTE} satisfies the Jacobi identity. If a basis is chosen such that the optical rotation term $\MN$ does not vanish, we were able to derive several almost-Poisson brackets, which are not shown here, i.e., antisymmetric bilinear forms that do not satisfy the Jacobi identity. Therefore, we postulate the following conjecture:
\begin{conjecture}
    It is not possible to construct a metriplectic formulation of the radiative transfer equation in any basis of the polarization plane that does not eliminate the optical rotation term.
\end{conjecture}
A heuristic argument supporting this conjecture is Rytov's law \cite{Vinitskii1990}, which may be seen as a non-holonomic constraint, and that Hamiltonian systems with non-holonomic constraints can only be described by an almost-Poisson bracket that does not satisfy the Jacobi identity \cite{VanDerSchaft1994}.

In this manuscript, we considered the domain $\dom = \RR^3\times \RR^3\setminus\{0\}$ without any boundary conditions. In future work, boundary conditions, which are relevant for multi-physics problems or control applications, may be treated by considering the port-Hamiltonian formalism \cite{vanDerSchaft2007, VanderSchaft2014}.

Finally, let us mention that the metriplectic formulation can be used to construct discretizations that (automatically) preserve energy and dissipate entropy also on the discrete level, see \cite{Suzuki2016, Kraus2017, Hirvijoki2021} for works in this direction.

\ack
VB and MS acknowledge support by the Dutch Research Council (NWO) via the Mathematics Clusters grant no. 613.009.133.
\section*{Appendix}

\changecolor{black}\hypertarget{Proof_Matrix_Bracket}{\textbf{Proof of Proposition \ref{prop: Matrix_Bracket}:}}\changecolor{blue}
    \textit{For sufficiently regular $2\times 2$ matrix-valued functions $\MU,\MV,\MW$ the bracket $\cbra{\cdot}{\cdot}_M$ defined in \eqref{eq: Matrix_Bracket} is bilinear, anti-symmetric and satisfies the Jacobi identity.}
    \begin{align*}
        \cbra{\cbra{\MU}{\MV}_M}{\MW}_M+\cbra{\cbra{\MV}{\MW}_M}{\MU}_M+\cbra{\cbra{\MW}{\MU}_M}{\MV}_M = 0.
    \end{align*}
\begin{proof}
For sufficiently smooth matrix-valued functions $\MU, \MV$ let us recall the definition of the bracket $\cbra{\cdot}{\cdot}_M$ given in \eqref{eq: Matrix_Bracket}, i.e.,
    \begin{align*}
        \cbra{\MU}{\MV}_M = \sum_{j=0}^3\cbra{U_j}{V_j}\sigma_j.
    \end{align*}
Bilinearity and anti-symmetry of the bracket follow directly from bilinearity and anti-symmetry of the canonical Poisson-bracket $\cbra{\cdot}{\cdot}$ given in \eqref{eq: cbra}. The Jacobi identity follows, since
\begin{align*}
  &\left[\left[\MU,\MV\right]_{M}, \MW\right]_M
  +\left[\left[\MV,\MW\right]_{M}, \MU\right]_M
  +\left[\left[\MW,\MU\right]_{M}, \MV\right]_M\\
  &=\sum_{j=0}^3 \left( [[U_j,V_j], W_j] +[[V_j,W_j], U_j] + [[W_j,U_j], V_j]\right) \sigma_j= 0,
\end{align*}
where we used in the last step that the canonical Poisson bracket $[\cdot,\cdot]$ satisfies the Jacobi identity.

\end{proof}
\noindent \changecolor{black}\hypertarget{Proof_Symp_Jacobi}{\textbf{Complement to the proof of Proposition \ref{prop: Symp_Bracket}:}}\changecolor{blue}
    \textit{For sufficiently smooth functionals $\fa, \fb,\fc$ of $\MW$, the bracket defined in \eqref{eq: Symp_Bracket_RTE} satisfies the Jacobi identity.}
    \begin{align*}
        \pbra{\pbra{\fa}{\fb}}{\fc}+\pbra{\pbra{\fb}{\fc}}{\fa}+\pbra{\pbra{\fc}{\fa}}{\fb}=0.
    \end{align*}
    
\begin{proof}
We adapt the steps presented in \cite[pp. 571-572]{Abraham1988} to our situation.
We will denote the $L^2(\dom,\HH)$ inner product as $\langle\cdot,\cdot\rangle$; see also \eqref{eq: InnerProductMatrix}.
For the first derivative of a functional $\fa$ it holds that:
\begin{align*}
    D\fa[\MW](\MU) = \lim_{\varepsilon\to 0 } \frac{1}{\varepsilon}\left( \fa[\MW+\varepsilon\MU]-\fa[\MW]\right) = \left\langle\fracd{\fa}{\MW}[\MW],\MU\right\rangle.
\end{align*}
We recall also the definition of the second derivative of a functional $\fa$, denoted $D^2\fa$.
For fixed $\MW$, $D^2\fa[\MW]$ is a bilinear map in perturbations $\MU,\MV$ of $\MW$ defined as the derivative of $D\fa[\MW](\MU)$, i.e.,
\begin{align*}
    D^2\fa[\MW](\MU,\MV)&=\lim_{\varepsilon\to 0} \frac{1}{\varepsilon} \left(D\fa[\MW+\varepsilon\MV](\MU)-D\fa[\MW](\MU)\right)\\
    &= \lim_{\varepsilon\to 0} \frac{1}{\varepsilon}\left( \left\langle \fracd{\fa}{\MW}[\MW+\varepsilon \MV],\MU \right\rangle - \left\langle \fracd{\fa}{\MW}[\MW],\MU\right\rangle\right)\\
    &= \left\langle \fracdsqr{\fa}{\MW}[\MW](\MV),\MU\right\rangle
\end{align*}
That is, $\fracdsqr{\fa}{\MW}[\MW](\MV)\in L^2(\dom,\HH)$ represents $D^2\fa[\MW](\cdot,\MV)$.
For simplicity, we will henceforth drop the explicit dependency on $\MW$, i.e., we write $\fracdsqr{\fa}{\MW}(\MV)$ instead of $\fracdsqr{\fa}{\MW}[\MW](\MV)$.

The derivative of the Poisson bracket evaluated in $\MW$ in direction $\MV$ is computed using the product rule
\begin{align*}
   &\left\langle\fracd{ \pbra{\fa}{\fb} }{\MW} ,\MV\right\rangle =D \pbra{\fa}{\fb}(\MV) =  D \left\langle \MW,\cbra{\fracd{\fa}{\MW}}{\fracd{\fb}{\MW}}_M \right\rangle(\MV)\\
    &=\left\langle \MV,\cbra{\fracd{\fa}{\MW}}{\fracd{\fb}{\MW}}_M \right\rangle
    +\left\langle \MW,\cbra{\fracdsqr{\fa}{\MW}(\MV)}{\fracd{\fb}{\MW}}_M \right\rangle
    +\left\langle \MW,\cbra{\fracd{\fa}{\MW}}{\fracdsqr{\fb}{\MW}(\MV)}_M \right\rangle.
\end{align*}
In order to simplify the second and the third term, we introduce the unbounded adjoint operator ${\rm ad}(\MX)$ defined by ${\rm ad}(\MX)[\MY]=\cbra{\MX}{\MY}_M$ and its dual operator ${\rm ad}(\MX)^*$, defined by
\begin{align*}
    \langle {\rm ad}(\MX)^* \MU,\MY\rangle = \langle \MU, \cbra{\MX}{\MY}_M\rangle,
\end{align*}
for sufficiently smooth $\MY,\MU$. Using the adjoint and its dual, as well as the definition of the second functional derivative, we obtain the following identities
\begin{align*}
    \left\langle \MW,\cbra{\fracdsqr{\fa}{\MW}(\MV)}{\fracd{\fb}{\MW}}_M \right\rangle &=- \left\langle {\rm ad}\left( \fracd{\fb}{\MW}\right)^*\MW, \fracdsqr{\fa}{\MW}(\MV) \right\rangle=- D^2\fa\left(\MV, {\rm ad}\left( \fracd{\fb}{\MW}\right)^*\MW\right),\\
    \left\langle \MW,\cbra{\fracd{\fa}{\MW}}{\fracdsqr{\fb}{\MW}(\MV)}_M \right\rangle &= \hphantom{-} \left\langle {\rm ad}\left( \fracd{\fa}{\MW}\right)^*\MW, \fracdsqr{\fb}{\MW}(\MV) \right\rangle =  \hphantom{-} D^2\fb\left( \MV, {\rm ad}\left( \fracd{\fa}{\MW}\right)^*\MW\right),
\end{align*}
where we also used the symmetry of the second derivative.
Collecting the terms, we identify
\begin{align*}
    \fracd{ \pbra{\fa}{\fb} }{\MW} = \cbra{\fracd{\fa}{\MW}}{\fracd{\fb}{\MW}}_M - \fracdsqr{\fa}{\MW}\left( {\rm ad}\left( \fracd{\fb}{\MW}\right)^*\MW\right) + \fracdsqr{\fb}{\MW}\left( {\rm ad}\left( \fracd{\fa}{\MW}\right)^*\MW\right).
\end{align*}
Consequently, we can evaluate
\begin{align*}
    &\hphantom{=}\pbra{\pbra{\fa}{\fb}}{\fc} = \left\langle \MW, \cbra{ \fracd{ \pbra{\fa}{\fb} }{\MW}}{\fracd{\fc}{\MW}}_M \right\rangle\\
    &=  \left\langle \MW, \cbra{ \cbra{\fracd{\fa}{\MW}}{\fracd{\fb}{\MW}}_M}{\fracd{\fc}{\MW}}_M \right\rangle
    -   \left\langle \MW, \cbra{ \fracdsqr{\fa}{\MW}\left( {\rm ad}\left( \fracd{\fb}{\MW}\right)^*\MW\right)}{\fracd{\fc}{\MW}}_M \right\rangle\\
    &\hphantom{=\ \left\langle \MW, \cbra{ \cbra{\fracd{\fa}{\MW}}{\fracd{\fb}{\MW}}_M}{\fracd{\fc}{\MW}}_M \right\rangle}+   \left\langle \MW, \cbra{ \fracdsqr{\fb}{\MW}\left( {\rm ad}\left( \fracd{\fa}{\MW}\right)^*\MW\right)}{\fracd{\fc}{\MW}}_M \right\rangle\\
    &=  \left\langle \MW, \cbra{ \cbra{\fracd{\fa}{\MW}}{\fracd{\fb}{\MW}}_M}{\fracd{\fc}{\MW}}_M \right\rangle
    +   \left\langle  {\rm ad}\left( \fracd{\fc}{\MW}\right)^*\MW, \fracdsqr{\fa}{\MW}\left( {\rm ad}\left( \fracd{\fb}{\MW}\right)^*\MW\right) \right\rangle\\
    &\hphantom{=\ \left\langle \MW, \cbra{ \cbra{\fracd{\fa}{\MW}}{\fracd{\fb}{\MW}}_M}{\fracd{\fc}{\MW}}_M \right\rangle}-    \left\langle {\rm ad}\left( \fracd{\fc}{\MW}\right)^* \MW, \fracdsqr{\fb}{\MW}\left( {\rm ad}\left( \fracd{\fa}{\MW}\right)^*\MW\right) \right\rangle\\
    &=\left\langle \MW, \cbra{ \cbra{\fracd{\fa}{\MW}}{\fracd{\fb}{\MW}}_M}{\fracd{\fc}{\MW}}_M \right\rangle
      + D^2\fa\left( {\rm ad}\left(\fracd{\fb}{\MW}\right)^*\MW, {\rm ad}\left(\fracd{\fc}{\MW}\right)^*\MW\right)\\
      &\hphantom{=\ \left\langle \MW, \cbra{ \cbra{\fracd{\fa}{\MW}}{\fracd{\fb}{\MW}}_M}{\fracd{\fc}{\MW}}_M \right\rangle}- D^2\fb\left({\rm ad}\left(\fracd{\fa}{\MW}\right)^*\MW, {\rm ad}\left(\fracd{\fc}{\MW}\right)^*\MW\right)
\end{align*}
Summing up the cycling permutations of the previous expression and using the fact that $\cbra{\cdot}{\cdot}_M$ satisfies the Jacobi identity, verifies the Jacobi identity for the Poisson bracket.
\end{proof}
\noindent \changecolor{black}\hypertarget{Proof_Casimirs}{\textbf{Proof of proposition \ref{prop: Casimirs}:}}\changecolor{blue}
\textit{The entropy functional $\fs$ defined in \eqref{eq: Entropy_Functional_RTE} is a Casimir of the Poisson bracket given in \eqref{eq: Symp_Bracket_RTE} and the Hamiltonian functional $\fh$ in \eqref{eq: Hamiltonian_Functional_RTE} is a Casimir of the metric bracket in \eqref{eq: Metr_Bracket_RTE}.}
\begin{proof}
    To show that $\fs$ is a Casimir invariant of the Poisson bracket, note that $\fracd{\fs}{\MW} = -\MW$ and take an arbitrary functional $\ff_\MU[\MW]$ with functional derivative $\fracd{\ff_\MU}{\MW} = \MU$.
    The Poisson bracket acting on these two functionals is then given by
    \begin{align*}
        -\pbra{\ff_\MU}{\fs} &= \int_\dom\Tr\left(\MW\cbra{\MU}{\MW}_M\right)\dx\dk\\
        &=\int_\dom\Tr\left(\sum_{i=0}^3\sum_{j=0}^3 W_{i}\cbra{U_{j}}{W_{j}}\sigma_i\sigma_j\right)\dx\dk\\
        &=\sum_{j=0}^3\int_\dom\left(W_j\cbra{U_j}{W_j}\right)\dx\dk\\
        &=\sum_{j=0}^3\int_\dom\left(U_j\cbra{W_j}{W_j}\right)\dx\dk = 0,
    \end{align*}
    where the last step follows from integration by parts, and the bracket vanishes due to
    the anti-symmetry of the canonical  Poisson bracket.

    To show that $\fh$ is a Casimir invariant of the metric bracket, note that $\fracd{\fh}{\MW}=\Omega$. Which depends on $k$ only as $|k|$. Since in the metric bracket, these functional derivatives are integrated over the sphere $|k|=|k'|$ these terms cancel and the metric bracket vanishes.
\end{proof}
\section*{References}
\bibliography{references}
\bibliographystyle{iopart-num}
\end{document}